\documentclass[a4paper,UKenglish,cleveref, autoref, thm-restate]{lipics-v2021}
\pdfoutput=1 %uncomment to ensure pdflatex processing (mandatatory e.g. to submit to arXiv)
\hideLIPIcs  %uncomment to remove references to LIPIcs series (logo, DOI, ...), e.g. when preparing a pre-final version to be uploaded to arXiv or another public repository

\usepackage{amsfonts}       % blackboard math symbols
\usepackage{nicefrac}       % compact symbols for 1/2, etc.
\usepackage{microtype}      % microtypography

\usepackage[table,xcdraw]{xcolor}
\usepackage{amsfonts}
\usepackage[misc]{ifsym}
\usepackage{amsmath}
\usepackage{booktabs}
\usepackage{amsthm}
\usepackage{mathtools}
\usepackage{thmtools}
\usepackage[ruled,linesnumbered,noend]{algorithm2e}
\usepackage{lipsum,graphicx,subcaption}
\newtheorem{problem}{Problem}

\newtheorem*{pfs*}{Proof Sketch}
\newtheorem*{proof*}{Proof}
\usepackage[noend]{algorithmic}

\newcommand{\bag}[1]{\texttt{bag}(#1)}
\newcommand{\true}{\texttt{True}}

\newcommand{\tw}{\texttt{tw}}
\newcommand{\false}{\texttt{False}}
\DeclarePairedDelimiter{\ceil}{\lceil}{\rceil}

\usepackage{tcolorbox}
\usepackage{physics}

\title{Quantum-Inspired Perfect Matching under Vertex-Color Constraints}

\titlerunning{Quantum-Inspired Perfect Matching under Vertex-Color Constraints} %TODO optional, please use if title is longer than one line

\author{Moshe {Y. Vardi}}{Department of Computer Science, Rice University, USA}{}{}{}%TODO mandatory, please use full name; only 1 author per \author macro; first two parameters are mandatory, other parameters can be empty. Please provide at least the name of the affiliation and the country. The full address is optional. Use additional curly braces to indicate the correct name splitting when the last name consists of multiple name parts.

\author{Zhiwei Zhang \footnote{The author list has been sorted alphabetically by last name; this should not be used to determine the extent of authors’ contributions. Corresponding author: Zhiwei Zhang (zhiwei@rice.edu).}} {Department of Computer Science, Rice University, USA}{zhiwei@rice.edu}{}{}
\authorrunning{M. Vardi and Z. Zhang} %TODO mandatory. First: Use abbreviated first/middle names. Second (only in severe cases): Use first author plus 'et al.'

\funding{Work supported in part by NSF grants IIS-1527668, CCF-1704883, IIS-1830549, DoD MURI grant N00014-20-1-2787,
Andrew Ladd Graduate Fellowship of Rice Ken Kennedy Institute, and an award from the Maryland Procurement
Office.}%optional, to capture a funding statement, which applies to all authors. Please enter author specific funding statements as fifth argument of the \author macro.


 %TODO mandatory; please add comma-separated list of keywords

\category{} %optional, e.g. invited paper

\relatedversion{} %optional, e.g. full version hosted on arXiv, HAL, or other respository/website
\nolinenumbers %uncomment to disable line numbering

\EventNoEds{2}
\EventShortTitle{Arxiv}
\EventAcronym{Arxiv}
\SeriesVolume{42}
\ArticleNo{}
\begin{document}

\maketitle

%TODO mandatory: add short abstract of the document
\begin{abstract}
We propose and study the graph-theoretical problem EXISTS-PMVC: \emph{the existence of perfect matching under vertex-color constraints on graphs with bi-colored edges}. EXISTS-PMVC is of special interest because of its motivation from quantum-state identification and quantum-experiment design, as well as its rich expressiveness, i.e., EXISTS-PMVC naturally subsumes important constrained matching problems, such as exact perfect matching. We give complexity and algorithmic results for EXISTS-PMVC under two types of vertex color constraints:   (1) decision-diagram constraints (EXISTS-PMVC-DD) and (2) symmetric constraints (EXISTS-PMVC-Sym).

For EXISTS-PMVC-DD, we reveal its NP-hardness by a graph-gadget technique. We prove that EXISTS-PMVC-Sym with a bounded number of colors (EXISTS-PMVC-Sym-Bounded) is polynomially equivalent with Exact Perfect Matching (XPM), which implies that EXISTS-PMVC-Sym-Bounded is in RNC on general graphs and PTIME on planar graphs. Directly applying algorithms for XPM to solve EXISTS-PMVC-Sym-Bounded is, however, impractical. We propose algorithms that natively handle EXISTS-PMVC-Sym-Bounded with considerably better complexity. Our novel results for EXISTS-PMVC provide insights into both constrained matching and scalable quantum experiment design.
\end{abstract}

\section{Introduction}
Quantum computing is widely believed to be one of the most promising research areas that could be a game changer, as it studies a radical new type of computational model \cite{deutsch1985quantum}.  Using quantum computing, it is hoped to tackle classically hard problems that are believed to be intractable for traditional computers, such as factorization by Shor's Algorithm \cite{shor1994algorithms,simon1997power}. Nevertheless, implementations of quantum computers today have not yet reached satisfactory scalability and reliability to solve practical problems. On the path of designing useful quantum computers, quantum experiments \cite{pan2012multiphoton} are necessary. Designing quantum experiments can be often related to solving classical problems, which forms a recently active research direction \cite{jooya2016graph,duncan2020graph}. In this direction, a surprising new connection has been recently discovered between the quantum state of a quantum optical circuit and perfect matchings (PM) with respect to vertex-color constraints on graphs with bi-colored edges \cite{krenn2017quantum}. For example, designing an optical circuit whose superposition exhibits GHZ state is equivalent to constructing a graph where the vertex colorings of all its perfect matchings form the set of all mono-chromatic colorings \cite{gu2019quantum3}. Studying perfect matchings under vertex-color constraints provides an important understanding of the properties of quantum circuits \cite{krenn2019questions}. 

Under the motivation above, we propose and investigate a quantum-inspired graph-theoretical problem of deciding whether there \textbf{exists} a \textbf{P}erfect \textbf{M}atching under \textbf{V}ertex-\textbf{C}olor constraints in a graph with bi-colored edges (EXISTS-PMVC). The original graph-theoretical problem proposed in \cite{krenn2019questions} involves operations of complex numbers, which prevents positive complexity and algorithmic results. PMVC is a simplified decision problem that corresponds to the case of the original problem when all edge weights are real and positive. %\myv{Why is this a special case?} zhiwei: I added some sentences but explaining this thoroughly can take more text 
This simplified problem is still of interest in quantum computing in that graphs with real, positive weights correspond to circuits with no phase-shift gates \cite{barenco1995elementary}. Designing a circuit satisfying an arbitrary quantum state with no phase-shift gate remains challenging \cite{krenn2019questions,krenn2017quantum}.

Studying EXISTS-PMVC is also valuable in graph theory, besides its motivation in quantum computing. Deciding whether perfect matchings exist in a graph is well-known to be in PTIME via the Blossom Algorithm \cite{edmonds1965paths}. Additional restrictions on perfect matchings often make the problem display interesting complexity and algorithmic properties \cite{elmaalouly2022exact}. While existing work studied PMs under edge color or weight constraints, not much effort has been put into PM under vertex-color constraints on graphs with bi-colored edges.

We point out that the complexity of EXISTS-PMVC depends heavily on the specific representation of the vertex-color constraints. We study two constraint representations of EXISTS-PMVC with a bounded number of colors: (1) decision diagrams (EXISTS-PMVC-DD-Bounded), and (2)  symmetric constraints (EXISTS-PMVC-Sym-Bounded). We prove by graph-gadget methods that the problem (EXISTS-PMVC-DD-Bounded) becomes NP-hard when the constraints are represented by decision diagrams. Despite the negative results on decision diagrams, most interesting quantum states yield constraints that do not need decision diagrams but only symmetry constraints. When the constraints are symmetric, we show that the problem with a bounded number of colors (EXISTS-PMVC-Sym-Bounded) is polynomially equivalent to the well-known problem \emph{Exact Perfect Matching} (XPM). This equivalence indicates that EXISTS-PMVC-Sym-Bounded is in RP and, in fact, RNC \cite{mulmuley1987matching}. %\myv{references?} %\myv{Next sentence is hard to parse.} zhiwei: rewritten
Compared with XPM, EXISTS-PMVC-Sym-Bounded seems to be a more general problem, in that it allows more types of constraints, bi-colored edges, and more than two colors. Our results, however, reveal that the generality of EXISTS-PMVC-Sym-Bounded does not lift the complexity. 

Towards the goal of efficiently solving quantum-inspired problems, we devise specific algorithms for EXISTS-PMVC-Sym-Bounded. The equivalence between this problem and XPM makes it possible to use XPM algorithms for solving EXISTS-PMVC-Sym-Bounded. 
 %\myv{Not clear} zhiwei: rewritten
 Such algorithms, however, are prohibitively expensive and impractical, as the reduction from EXISTS-PMVC-Sym-Bounded to XPM greatly increases the size of the graph. 
 %\myv{Why?} zhiwei: explained. 
 We adapt algorithms of XPM based on symbolic determinant \cite{mulmuley1987matching} to solve EXISTS-PMVC-Sym-Bounded with better complexity. We also devise dynamic-programming (DP)-based algorithm for efficiently solving instances with low-width problems.  Our results not only fill the knowledge gap in constrained perfect matching under vertex-related constraints and bi-colored graphs, but also provide algorithmic progress in identifying properties of large-scale quantum circuits.
 
%(optional) We implement four types of methods (enumerate-Blossom, DP, symbolDet, propositional logic) on two types of benchmarks motivated by GHZ and Dicke. We show that different approaches apply on different methods and there is no clear winner.

%\begin{tcolorbox} 
The main contributions of this paper are ($n$ is the number of vertices, $d$ is the number of colors and $\tw$ is the graph treewidth):

\begin{bracketenumerate}
\item  EXISTS-PMVC-DD with at least two colors is  NP-hard.
\item EXISTS-PMVC-Sym-Bounded is polynomially equivalent to XPM. 
\item EXISTS-PMVC-Sym-Bounded can be solved w.h.p. in $O(n^{2d+1})$ by a symbolic-determinant technique.
\item EXISTS-PMVC-Sym-Bounded can be solved deterministically in $O(n^{2d}\cdot (2d+2)^{\tw+1})$ by dynamic programming.
%\item An experimental comparison of different algorithms.
\end{bracketenumerate}
%\end{tcolorbox}

\iffalse
\subsection{Paper Organization}
Section \ref{sec:pre} presents preliminaries. Section \ref{section:quantumInterpretation} describes the motivation of EXISTS-PMVC from quantum computing.  In Section \ref{sec:PMSym} we give positive complexity and algorithmic results for EXISTS-PMVC-Sym. In  Section \ref{sec:PMVCDD} we prove the NP-hardness of EXISTS-PMVC-DD. Section \ref{sec:conlusion} concludes the paper and lists open problems and future directions. Due to space limit, most proofs of the results in this paper are deferred until the appendix.
\fi

\section{Problem Definitions and Preliminaries}
\label{sec:pre}

\subsection{Complexity Classes}
Let PTIME be the complexity class of polynomially solvable problems. RP is the class of problems where a polynomial randomized algorithm $\mathcal{R}$ exists such that 1) if the actual answer is ``yes'', then $\mathcal{R}$ returns ``yes'' with constant non-zero probability;  2) if the answer is ``no'', $\mathcal{R}$ always returns ``no''.  NC is the class of problems solvable by a parallel algorithm in logarithmic time with polynomially many processors. NC$^i$ is the subset of NC of problems solvable in $O(\log^i n)$ time by polynomially many processors. RNC and RNC$^i$ are the randomized analogs of NC and NC$^i$, respectively. We have NC$^1\subseteq$ NC$^2 \subseteq\cdots \subseteq$ NC$\subseteq P$ and RNC$^1\subseteq$ RNC$^2 \subseteq\cdots \subseteq$ RNC$\subseteq$ RP.  For two problems $\mathcal{P}_1, \mathcal{P}_2$,  we have that $\mathcal{P}_1\le_P \mathcal{P}_2$ if there is a polynomial reduction from $\mathcal{P}_1$  to $\mathcal{P}_2$. We denote $\mathcal{P}_1<_P \mathcal{P}_2$ if $\mathcal{P}_1\le_P \mathcal{P}_2$ and $\mathcal{P}_2\not\le_P \mathcal{P}_1$. Finally, $\mathcal{P}_1=_P\mathcal{P}_2$ (polynomial equivalence) if $\mathcal{P}_1\le_P \mathcal{P}_2$ and $\mathcal{P}_2\le_P \mathcal{P}_1$.

\subsection{Perfect Matchings in Graphs with Bi-colored Edges under Vertex-Color Constraints}

\begin{figure}[th!]
	\centering
\includegraphics[scale=0.2]{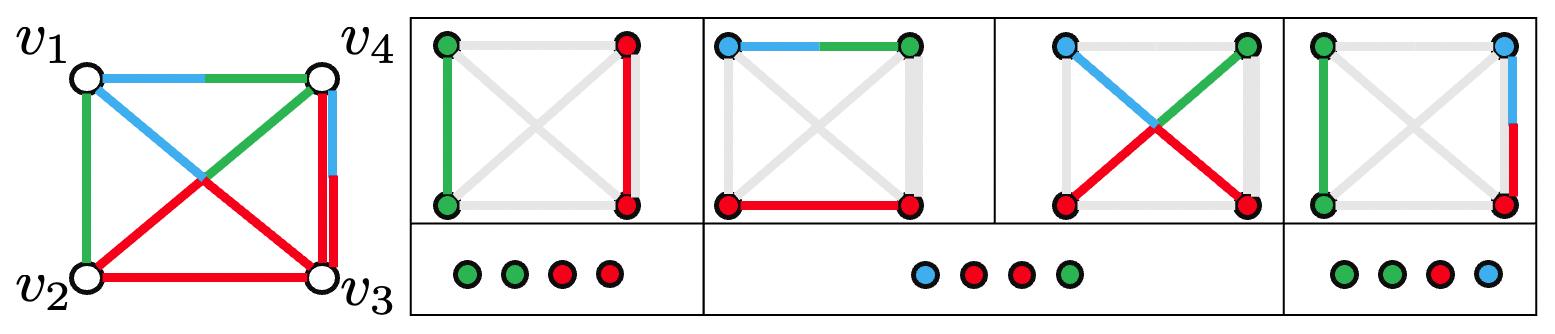}
\caption{Left: a graph $G$ with bi-colored edges. Right: all perfect matchings of $G$ and their inherited vertex colorings. }
\label{fig:graph}
\end{figure}

\begin{definition}
{\rm (Graphs with bi-colored edges)} In this paper, a ``graph'' refers to an undirected, non-simple \footnote{Self-loops are prohibited but multiple edges between the same pair of vertices are allowed.} graph with bi-colored edges. Formally, a graph $G$ is a tuple $(V,E,d)$, where $V$ is the set of vertices with $|V|=n$, $E$ is the set of bi-colored edges and $d$ is the number of colors. For each edge $e\in E$ that connects $u$ and $v$, there are two colors $c_u^e,c_v^e\in \{1,\cdots, d\}$ that are associated with $u$ and $v$, respectively, i.e., $e=\{(u,c_u^e),(v,c_v^e)\}$. If $c_u^e=c_v^e$, we call the edge $e$ monochromatic, otherwise $e$ is bi-colored. See Figure \ref{fig:graph} (left) as an example.
 \end{definition} 

\begin{definition}
{\rm (Vertex coloring)} A vertex coloring of graph $G=(V,E,d)$ is a mapping $c$ from vertices to color indices i.e., $c:V\to \{1,\cdots,d\}$. Alternatively one can write  $c\in {\{1,\cdots,d\}}^V$.
\end{definition} 

\begin{definition}
{\rm (Perfect matchings and their inherited vertex colorings) \cite{krenn2019questions}} A perfect matching $P$ of a graph $G$ is a subset of edges, i.e., $P\subseteq E$, such that for each vertex $v\in V$, exactly one edge in $P$ is adjacent to $v$. For a perfect matching $P$ of a graph $G$, its \textbf{inherited vertex coloring}, denoted by $c^P:V\to \{1,\cdots, d\}$ is defined as follows: for each vertex $v$, let $e$ be the unique edge that is adjacent to $v$, then $c^P(v)=c_v^e$. See Figure \ref{fig:graph} (right) for examples. The color count of a perfect matching $P$ is a tuple $(n_1,\cdots,n_d)$ where $n_i$ is the number of vertices in color $i$ in $c^P$ for each $i\in \{1,\cdots,d\}$. 
\end{definition} 

\begin{definition} {\rm (Legal perfect matchings w.r.t. vertex-color constraints)} Suppose a system of vertex-color constraints defines a set of ``legal'' vertex colorings $\mathcal{C}\subseteq {\{1,\cdots,d\}}^V$. A perfect matching is  \emph{legal} (w.r.t. $\mathcal{C}$) if its inherited vertex coloring is legal, i.e., $c^P\in\mathcal{C}$.
\end{definition}

\begin{problem}
{\rm (Existence of \textbf{P}erfect \textbf{M}atching under \textbf{V}ertex-\textbf{C}olor constraints, EXISTS-PMVC)} Given a graph $G=(V,E,d)$  and a set $\mathcal{C}$ of legal vertex colorings defined by vertex-color constraints, does there exist a legal perfect matching of $G$ w.r.t. $\mathcal{C}$?
\label{problem:PMVC}
\end{problem} 

In the case of only polynomially many legal vertex colorings, e.g., ``all vertices have the same color'', the problem is easy, as shown in the following proposition. The idea is that for each vertex coloring $c$ in $\mathcal{C}$, we can call Blossom Algorithm on a subgraph of $G$ with only edges whose colors align with $c$. 
\begin{proposition}
If $|\mathcal{C}|$ is polynomially bounded by $|G|$, then EXISTS-PMVC is in PTIME. 
\label{prop:polynomiaClIsEasy}
\end{proposition} 

The challenging part of EXISTS-PMVC is when $|\mathcal{C}|$ is exponential, such that it is infeasible to explicitly enumerate $\mathcal{C}$. In this case, $\mathcal{C}$ needs to be defined implicitly by vertex-color constraints. We study two representations of vertex-color constraints: symmetric constraints and decision diagrams.

\subsection{Symmetric Constraints and EXISTS-PMVC-Sym}
\begin{definition}
{\rm(Symmetric constraints)  }
For a system of vertex-color constraints, let the vertex coloring set it defines be $\mathcal{C}$. The system of vertex-color constraints is \emph{symmetric} if  whether a vertex coloring $c$ is in $\mathcal{C}$ is determined by the count of each color $i\in\{1,\cdots,d\}$ in $c$, denoted by $\texttt{count}(c,i)=|\{v|c(v)=i,v\in V\}|$. 
\end{definition}
For example, ``there is only one red vertex'' ($\texttt{count}(c,\texttt{red})=1$) and ``all vertices must be blue'' ($\texttt{count}(c,\texttt{blue})=n$) are symmetric, while ``$v$ must be blue'' is asymmetric.
\begin{problem} (Existence of \textbf{P}erfect \textbf{M}atching under \textbf{Sym}metric \textbf{V}ertex-\textbf{C}olor constraints, EXISTS-PMVC-Sym)
 Does a graph have a legal perfect matching w.r.t. a symmetric system of constraints? When the number of colors is bounded, we denote the problem by EXISTS-PMVC-Sym-Bounded.
\label{problem:EXISTS-PMVC-Sym}
\end{problem}

\begin{figure}[ht]
	\centering
\includegraphics[scale=0.19]{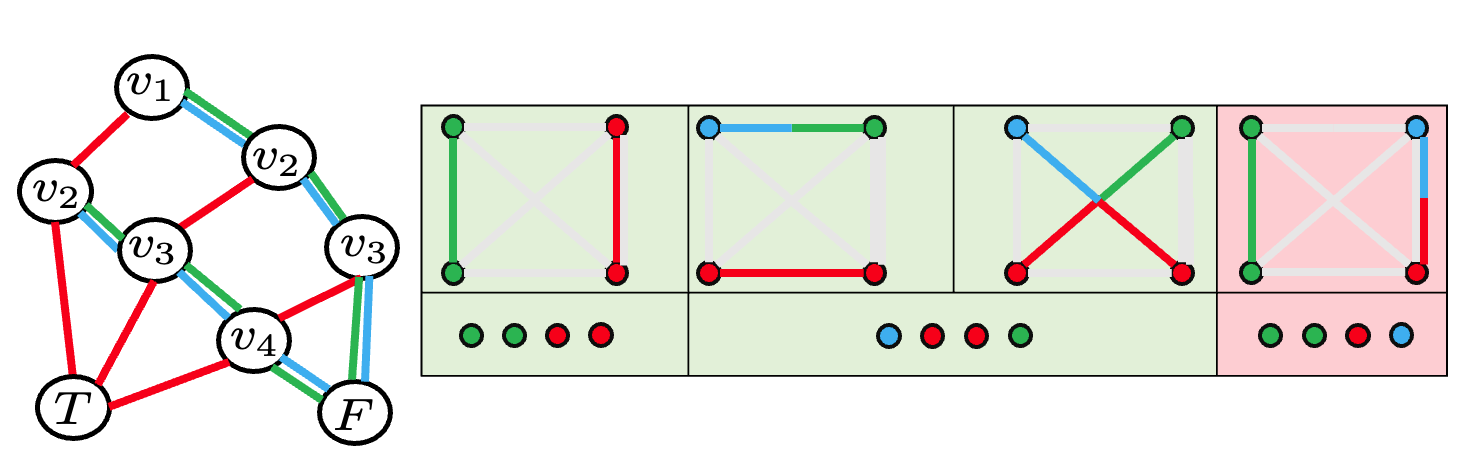}
\caption{Left: a DD encoding ``the number of red vertices is at least $2$'' on graphs with $4$ vertices. Right: for the graph in Figure \ref{fig:graph}, all perfect matchings except the last one are legal w.r.t. to the DD.}
\label{fig:dd}
\end{figure}

\subsection{Decision Diagram and EXISTS-PMVC-DD}
\begin{definition} {\rm (Decision Diagram)} 
A (reduced, ordered) Decision Diagram (DD), as a generalization of binary DD (BDD) \cite{bryant1992symbolic}, is a data structure that, in this paper, encodes a set of vertex colorings as a directed acyclic graph $D$. There are two terminal nodes with no outgoing edges, expressing \true\, and \false. Each non-terminal node is labeled with a vertex and also has $d$ outgoing edges. Each edge is labeled with a color index in $\{1,\cdots,d\}$. The DD is ordered, i.e., different vertices appear in the same order on all paths from the root.
\end{definition}

Whether a vertex coloring $c:V\to \{1,\cdots,d\}$ is legal w.r.t. the DD can be determined by DD evaluation. The evaluation starts from the unique root of the DD and at each non-terminal node labeled by $v\in V$, goes through the edge labeled by $c(v)$ to the next node, until a terminal node is reached. The coloring is legal iff. the reached terminal node is \true.

\begin{problem} (Existence of \textbf{P}erfect \textbf{M}atching under \textbf{D}ecision \textbf{D}iagram \textbf{V}ertex-\textbf{C}olor constraints, EXISTS-PMVC-DD)
 Does a graph  have a legal perfect matching w.r.t. a  decision diagram that encodes legal vertex colorings? When the number of colors is bounded, we denote the problem by EXISTS-PMVC-DD-Bounded.
\label{problem:EXISTS-PMVC-DD}
\end{problem} 

EXISTS-PMVC-DD subsumes EXISTS-PMVC-Sym, because systems of symmetric constraints always admit polynomial size DDs \cite{kyrillidis2021continuous}. % 2) polynomial size DDs can encode asymmetric constraints.
%\begin{proposition}
%EXISTS-PMVC-Sym-Bounded $<_P$ EXISTS-PMVC-DD-Bounded.
%\end{proposition}
When the number of colors is unbounded, both EXISTS-PMVC-Sym and EXISTS-PMVC-DD become NP-complete. The results can be obtained by reducing rainbow perfect matching (RPM) to EXISTS-PMVC-Sym and EXISTS-PMVC-DD. Given a graph with $n/2$ edge colors, RPM asks if there is a PM with all different colors. RPM is NP-hard since it can be reduced from 3-dimensional perfect matching (3DM), an NP-hard problem.

%\begin{proposition}
%EXISTS-PMVC-Sym and EXISTS-PMVC-DD with an unbounded number of colors are NP-Complete.
%\end{proposition}

In the rest of this paper, we focus on EXISTS-PMVC with a bounded number of colors.

\begin{figure}[ht]
	\centering
\includegraphics[scale=0.2]{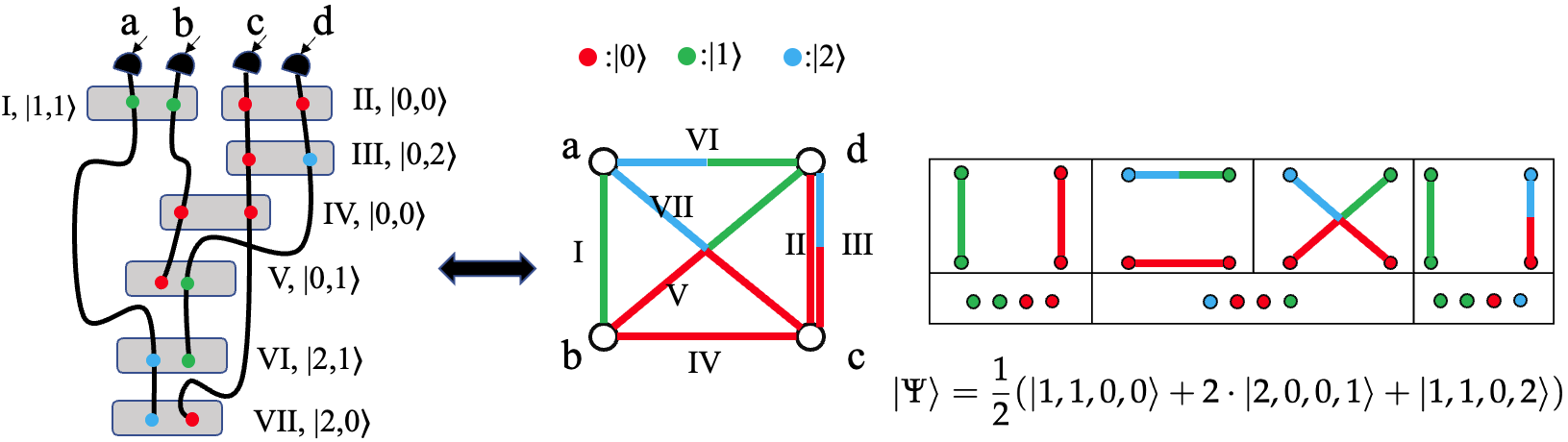}
\caption{Left: a quantum optical circuit corresponding to the graph in Figure \ref{fig:graph}. Each crystal (shaded box) correspond to an edge in graph. Each optical path is converted into a vertex in graph. The modes of photons translate to edge colors. Right: all PMs in the graph reflect the quantum state of the optical circuit as a superposition of coincidences.  \cite{gu2019quantum2}}
\label{fig:circuitPM}
\end{figure}

\section{Quantum Motivation of EXISTS-PMVC}
\label{section:quantumInterpretation}
This section describes the motivation of EXISTS-PMVC from the perspective of quantum-experiment design. An example of a (simplified) \emph{quantum optical circuit} is shown in Figure \ref{fig:circuitPM} (left). A non-linear \emph{crystal} (shaded box in Figure \ref{fig:circuitPM}) can emit a pair of entangled \emph{photons} simultaneously if activated by a laser-pump power. Each photon has a \emph{mode}, encoded by an integer index. Photons emitted by crystals travel on  \emph{optical paths}. A \emph{receiver}, which detects the arrival of a photon and identifies photon mode, is placed at the end of each optical path. Each crystal is associated with two optical paths \cite{gu2019quantum2}. % Each crystal also has a (complex) weight representing its amplitude  \cite{gu2019quantum2}.

A \emph{coincidence} of a quantum optical circuit happens when each receiver detects exactly one photon, due to the activation of some crystals. The \emph{quantum state} of a coincidence is the mode of photons caught by all receivers in the coincidence. The quantum state of an optical circuit is the \emph{superposition} \cite{ballentine2014quantum} of the quantum states of all coincidences. 

The optical circuit in Figure \ref{fig:circuitPM} (left) has $7$ nonlinear crystals. A coincidence happens only when crystals in \{I,II\}, or \{I,III\}, or \{IV,VI\} or \{V,VII\} are activated simultaneously.

In \cite{gu2019quantum2}, a coincidence of an optical circuit is shown to be equivalent to a perfect matching under vertex-color constraints in an undirected graph with bi-colored edges. The graph $G$ can be constructed as follows. See Figure \ref{fig:circuitPM} as an example.
\begin{bracketenumerate}
    \item Each optical path corresponds to a vertex in $G$.
    \item Each crystal corresponds to an edge in $G$. This edge connects two vertices corresponding to two optical paths which the crystal can emit photons to.
    \item The modes of photons emitted by crystals correspond to edge colors. Since a crystal can emit two photons with different modes, edges in $G$ are bi-colored.
    \item A coincidence of an optical circuit corresponds to a perfect matching $P$ in $G$. The set of edges in $P$ indicates activated crystals. The quantum state of the coincidence corresponds to the inherited vertex coloring of $P$.
    \item Each crystal has an amplitude as a complex number, corresponding to the edge weight in graph. The amplitude of a coincidence is the product of the amplitudes of all activated crystals. The weight of a perfect matching is the product of weights of all its edges. 
\end{bracketenumerate}

A quantum state of an optical circuit can be viewed as a superposition of coincidences with predefined ``legal'' quantum states. An optical circuit exhibits a quantum state if the total amplitude of all coincidences of each legal states is $1$, while all coincidences with an ``illegal'' state have total amplitude $0$ \cite{krenn2017quantum}. In the case where only real positive amplitudes are allowed, a necessary condition of exhibiting a quantum state is the non-existence of illegal coincidences. In the graph-theoretical setting, the necessary condition translates to the non-existence of illegal PMs, which can be refuted by EXISTS-PMVC algorithms. We list quantum states of interest and their corresponding vertex-color constraints in Table \ref{table:quantumState}.

	\begin{table}[ht]
		\centering
		\begin{small}
		\begin{tabular}{c c c c c}
			\toprule 
			Name \cite{gu2019quantum3}  & &  Quantum state & &  Vertex-color constraints  \\
			\cmidrule{1-1} \cmidrule{3-3} \cmidrule{5-5}
			GHZ State & & ${1}/{\sqrt{d}}\cdot\sum_{i=0}^{d-1}\ket{i}^{\oplus n}$ & & $\bigvee_{i=1}^d(\texttt{count}(c,i)=n)$ \vspace{0.2cm}   \\ 
			Dicke State && ${1}/{\sqrt{\binom{n}{k}}}\cdot\hat{S}(\ket{0} ^{\oplus (n-k)}\ket{1} ^{\oplus k})$ &&  $ \texttt{count}(c,1)=n-k$ \\ 
				W State && ${1}/{\sqrt{n}}\cdot\hat{S}(\ket{0} ^{\oplus (n-1)}\ket{1} ^{\oplus 1})$ &&  $ \texttt{count}(c,1)=n-1$   \\ 
					General Dicke State && $1/{\sqrt{\binom{n}{k_0,\cdots, k_{d-1}}}}\cdot\hat{S}(\ket{0}^{\oplus k_0}\cdots \ket{d-1}^{\oplus k_{d-1}})$ && $\bigwedge_{i=1}^d(\texttt{count}(c,i)=k_{i-1})$ \\ 
			\bottomrule
		\end{tabular} 
		\end{small}
		\caption{Quantum states and their corresponding vertex-color constraints on perfect matchings.}
		\label{table:quantumState}
	\end{table}
	
%\section{Related Work}
%\label{sec:relatedWork}
 The connection between constrained PM  and quantum computing has been proposed in a series of works \cite{gu2019quantum2,gu2019quantum3,krenn2019questions}. This line of research has inspired interesting theoretical open questions recently, e.g., Krenn's Conjecture on possible optical circuits that exhibit mono-chromatic GHZ state. While existing work focused on proving Krenn's Conjecture in  special cases \cite{chandran2022perfect} or verifying it up to a relatively small range \cite{cervera2021design}, investigation of complexity and algorithmic work for identifying other important quantum states besides GHZ, or arbitrary user-defined states  \cite{krenn2019questions} on large graphs is lacking.

\section{NP-Hardness of EXISTS-PMVC-DD-Bounded}
\label{sec:PMVCDD}
We prove that EXISTS-PMVC-DD with more than two colors is NP-hard, even for special classes of graphs listed in Corollary \ref{coro:specialGraph}. This result is interesting for two reasons. First, decision diagrams are considered tractable structures, whose satisfiability and counting are in polynomial time w.r.t. size of DD \cite{andersen1997introduction}. Our result reveals that when DD is combined with another polynomial problem, i.e., PM, the problem EXISTS-PMVC-DD becomes NP-hard. Second, the NP-hardness of EXISTS-PMVC-DD depends on neither exponential weighting nor an unbounded number of colors, unlike related constrained matching problems (see Remark \ref{remark:relatedProblems}) \cite{mastrolilli2014bi,busing2018budgeted}. We prove the following theorem by reducing 3-SAT to EXISTS-PMVC-DD with only red-blue mono-chromatic edges by the graph-gadget technique. 

\begin{theorem} EXISTS-PMVC-DD with number of colors $d\ge 2$ is NP-hard.
\label{theo:EXISTS-PMVC-BDD}
\end{theorem}

\begin{corollary}
   From the constructive proof of Theorem \ref{theo:EXISTS-PMVC-BDD},  EXISTS-PMVC-DD-Bounded on the following special classes of graphs is still NP-hard: 1) Graphs with bounded treewidth $\tw\ge 3$. 2) Non-simple bipartite graphs. 3) Graphs with only mono-chromatic edges.
   \label{coro:specialGraph}
      % \item is the problem on simple graph still NP-hard? Yes. The reduction can be modified into a simple graph version. Nevertheless, the graph is no longer bipartite. this should be written as a remark
\end{corollary}

\section{Algorithms for EXISTS-PMVC-Sym-Bounded}
We show that EXISTS-PMVC-Sym with a bounded number of colors is polynomially equivalent to XPM. As a result, EXISTS-PMVC-Sym-Bounded is in RP and RNC$^2$, which can also be solved deterministically in polynomial time on planar graphs by Pfaffian-orientation-based derandomization. Whether EXISTS-PMVC-Sym-Bounded is NP-hard is still open, same with XPM. The polynomial equivalence also suggests that one could use algorithms for XPM for solving  EXISTS-PMVC-Sym-Bounded. Nevertheless, the reduction from EXISTS-PMVC-Sym-Bounded to XPM comes with a significant overhead in the graph size. Therefore, we adapt symbolic-determinant algorithms for XPM to natively handle bi-colored graphs, with significantly better complexity compared with directly applying XPM algorithms. We also propose a dynamic-programming algorithm for graphs with bounded treewidth.
\label{sec:PMSym}

\subsection{EXISTS-PMVC-Sym-Bounded is Polynomially Equivalent to XPM}
We show the polynomial equivalence between EXISTS-PMVC-Sym-Bounded and a well-known graph theoretical problem called exact perfect matching (XPM).
\begin{problem} {\rm  \cite{elmaalouly2022exact} }
An Exact Perfect Matching (XPM)instance $(G,k)$ asks whether there exists a PM that contains exactly $k$ red edges in graph $G$ with blue/red-colored edges. 
\end{problem}
Strong evidence that XPM is not  NP-hard has been discovered, such as XPM being in RNC and having good approximation algorithms \cite{yuster2012almost}. Nevertheless, whether XPM is in P still remains one of the most intriguing open problems, standing open for more than $40$ years. %If one allows a real weight assigned to each edge and asks for a PM with total weight exactly $W$ (EWPM) \cite{gurjar2012planarizing}, then the problem is then NP-hard, though it can be reduced to XPM if all weights are polynomially bounded. 

  In practice, EXISTS-PMVC-Sym is more general than XPM in the sense that 1) there can be more than two colors and 2) bi-colored edges are allowed and 3) more types of constraints can be used.  Thus our result is interesting as it reveals that allowing bi-colored edges, increasing the number of colors from $2$ to a bounded number, and symmetric constraints do not lift the complexity of XPM. 

 \begin{theorem} EXISTS-PMVC-Sym-Bounded $=_P$ XPM.
 \label{conjecture:xpmAndSym}
\end{theorem}
\begin{proof}
\textbf{XPM $\le_P$ EXISTS-PMVC-Sym} with two colors and only monochromatic edges: The graph $G$ in EXISTS-PMVC-Sym remains the same. We only need to transform the cardinality constraint of edge colorings to a system of symmetric vertex-color constraints.  Since all edges are monochromatic, the set of perfect matchings with $k$ red edges equals the set of perfect matchings with $2k$ red vertices in their inherited vertex colorings. Thus the symmetric constraint for vertex coloring $c$ in EXISTS-PMVC-Sym is $\texttt{count}(c, \texttt{red})=2k$.

\textbf{EXISTS-PMVC-Sym-Bounded $\le_P$ XPM}: 
We first show that EXISTS-PMVC-Sym-Bounded on a graph $G$ with bi-colored edges can be reduced to EXISTS-PMVC-Sym-Bounded on a graph $G_1$ with only mono-chromatic edges. Given $G$ with $d$ colors, we construct such a $G_1$ with $d+1$ colors as follows. Each bi-colored edge $\{(u,c_u)(v,c_v)\}$ of $G$ is extended into a path in $G_1$ with length $3$: $\{(u,c_u),(u',c_u)\}$, $\{(u',d+1),(v',d+1)\}$, $\{(v',c_v),(v,c_v)\}$ (note that $d+1$ is a new color). Then we show that $G$ has a PM $P$ with color count $(n_1,\cdots,n_d)$ iff $G_1$ has a PM $P_1$ with color count $(n_1,\cdots,n_d, 2\cdot \#\texttt{BCE}(G))$, where $\#\texttt{BCE}(G)$ is the number of bi-colored edges in $G$. 
%The key observation is that there is a bijection between PMs in $G$ and PMs in $G_1$. Each perfect matching $P$ in $G$ corresponds to a unique perfect matching $P_1$ of $G_1$ and vice versa as follows. 
\begin{itemize}
    \item If $G$ has a PM $P$ with color count $(n_1,\cdots, n_d)$, then we construct a PM $P_1$ in $G_1$ as follows. For each bi-colored edge $e=\{(u,c_u),(v,c_v)\}\in P$, let $\{(u,c_u),(u_1,c_u)\}$ and $\{(v_1,c_v),(v,c_v)\}$ be in the PM $P_1$ of $G_1$. For each bi-colored edge $e=\{(u,c_u),(v,c_v)\}\not \in P$, let $\{(u_1,d+1),(v_1,d+1)\}$ be in $P_1$. For each monochromatic edge $e\in P$, we add $e$ to $P_1$.  It is easy to verify that the color count of $P_1$ is $(n_1,\cdots,n_d, 2\cdot \#\texttt{BCE}(G))$.
    \item If $G_1$ has a PM $P_1$ with color count $(n_1,\cdots,n_d, 2\cdot \#\texttt{BCE}(G))$,  then we construct a PM $P$ in $G$ as follows. For each edge $e\in P_1$, if both the vertices of $e$ are also in $G$, then we add $e$ to $P$. If $e$ is in the form of $\{(u,c_u),(u',c_u)\}$, then by the construction of $G_1$, $e$ is on a path of length $3$ in $G_1$: $\{(u,c_u),(u',c_u)\}$, $\{(u',d+1),(v',d+1)\}$, $\{(v',c_v),(v,c_v)\}$. We add the bi-colored edge $\{(u,c_u),(v,c_v)\}$ to $P$. The color count of $P$ is $(n_1,\cdots, n_d)$.
\end{itemize}
 
Next, we prove that EXISTS-PMVC-Sym-Bounded on a graph $G_1$ with only monochromatic edges and color count $(n_1,\cdots, n_d)$ can be reduced to an XPM instance $(G_2,\sum_{i=1}^{d}n_i\cdot (n/2)^{i-1})$. The graph $G_2$ is constructed as follows. Each edge of $G_1$ with color $i\in \{2,\cdots,d\}$ is extended to a path of length $2\cdot (n/{2})^{i-1}-1$ in $G_2$. The path is colored with  monochromatic red and blue edges alternately with $(n/2)^{i-1}$ red edges and $(n/2)^{i-1}-1$ blue edges, starting and ending with both red edges. Then we show that $G_1$ has a PM with color count $(n_1,\cdots,n_d)$ iff $G_2$ has a PM with $\sum_{i=1}^{d}n_i \cdot (n/2)^{i-1}$ red edges.

\begin{itemize}
    \item Suppose $G_1$ has a PM $P_1$ with color count $(n_1,\cdots, n_d)$. Then we construct a PM $P_2$ of $G_2$ as follows. For each edge $e\in P_1$ in $G_1$, put all red edges on the path that corresponds to $e$ in $G_2$ to $P_2$. For each edge $e\not\in P_1$, put all blue edges on the corresponding path in $G_2$ to $P_2$. Then one can verify that $P_2$ is a perfect matching of $G_2$ with $\sum_{i=1}^{d}n_i\cdot (n/2)^{i-1}$ red edges. 
    \item Suppose $G_2$ has a PM $P_2$ with $k=\sum_{i=1}^{d}n_i\cdot (n/2)^{i-1}$ red edges. We construct a PM $P_1$ as follows. By the construction of $G_2$, for each edge of $G_1$ with color $i$, there is a corresponding path with length $2\cdot (n/2)^{i-1}$ with $(n/2)^{i-1}$ red edges and $(n/2)^{i-1}-1$ blue edges in $G_2$. For each extended path in $G_2$, either all $(n/2)^{i-1}$ red edges all $(n/2)^{i-1}-1$ blue edges are in $P_2$. If all red edges in this path are in $P_2$, we include the corresponding edge in $G_1$ to $P_1$. It is clear that $P_1$ is a perfect matching of $G_1$. Next, we consider the color count  of $P_1$. Each extended path corresponding to an edge colored by $i$ in $G_1$ either contributes $(n/2)^{i-1}$ red edges or $0$  red edges to the PM $P_2$ in $G_2$. Since the count of each color in $P_1$ is bounded by $n/2$ and $k$ is bounded by $n/2\cdot (n/2)^{d-1}=(n/2)^d$, the color count of $P_1$ is given by the unique base-$(n/2)$ representation of $k$ with $d$ digits. i.e., $(n_1,\cdots,n_d)$ from the decomposition of $k=\sum_{i=1}^{d}n_i\cdot (n/2)^{i-1}$. 
\end{itemize}
Lastly, we show that both of the two reductions above are polynomial given $d$ is bounded. Let the vertex and edge set of $G$ be $V_G$ and $E_G$, respectively. Then the size of $G_1$ is bounded by $(|V_G|+2\cdot\#BCG(G), |E_G|+2\cdot \#BCG(G))$.  The size of $G_2$ is bounded asymptotically by $O(\frac{|E_{G_1}|}{d}\cdot (\frac{|V_{G_1}|}{2})^{d-1})=O((\frac{|V_G|+2\cdot |E_G|}{2})^{d-1}\cdot \frac{3\cdot |E_G|}{d})$. 
\end{proof} 

Since XPM is known to be in RNC$^2$ \cite{mulmuley1987matching}, EXISTS-PMVC-Sym-Bounded is also in RNC$^2$.  The RNC$^2$ algorithm of XPM uses $n^{3.5}$ processors. Thus directly applying the best-known algorithm for EXISTS-PMVC-Sym-Bounded can yield an RNC$^2$  algorithm with $O(n^{7d})$ processors (or $O(n^{7d}\cdot log^2n)$ time complexity for the sequential algorithm) in the worst case (on a dense graph where $|E|=O(|V|^2)$), which is prohibitively expensive. In the later part of this section, we propose two sequential algorithms for EXISTS-PMVC-Sym-Bounded that natively handles bi-colored graphs. The first is randomized, based on symbolic determinant with complexity $O(n^{2d+1})$ and the second is deterministic, based on dynamic programming with complexity $O(n^{2d}\cdot(2d+2)^{\tw +1})$.

\begin{remark} {\rm (Related problems)}
\label{remark:relatedProblems}
Bounded color matching (Bounded-CM)  \cite{mastrolilli2014bi} and a more general problem, budget colored matching (Budgeted-CM) \cite{busing2018budgeted} assign an integer cost and weight to each edge, and a budget to each color. The goal is to find a matching with maximum total weight within the edge color budget. Bounded-CM and Budgeted-CM are proved strongly NP-hard \cite{rusu2008maximum} if the number of colors is unbounded. Multi-budgeted matching \cite{busing2018multi} further generalizes Budgeted-CM by additionally assigning each edge a real vector cost and enforcing linear constraints on maximum matchings, which makes the problem  strongly NP-hard even for paths with an unbounded number of colors. 

EXISTS-PMVC subsumes cases of Bounded-CM and Budgeted-CM where all edges have unit weight and only perfect instead of maximum matchings are valuable. Note that in those related problems, either an unbounded number of colors or exponential weighting is necessary for proving NP-hardness, while EXISTS-PMVC-DD is NP-hard even with two colors. 
\end{remark}

\subsection{Algorithms Based on Symbolic Matrix Determinant} We adapt the symbolic-determinant approach for XPM to solve EXISTS-PMVC-Sym and obtain an RNC$^2$ algorithm. We further investigate the detailed complexity of our approach as a sequential algorithm via fraction-free Gaussian Elimination.
\label{subsection:RPAlgorithm}
\subsubsection{Tutte Matrix for PMs in Simple, Uncolored Graphs}
\label{subsubsec:tutte}
We recall the RP algorithm for perfect matching based on Tutte Matrix \cite{mulmuley1987matching}.

\begin{definition} {\rm (Tutte Matrix) \cite{tutte1947factorization}}The following skew-symmetric matrix $T$ is called the \emph{Tutte Matrix} of a simple, uncolored graph $G$, %\footnote{Here by ``graph'', we mean a simple, uncolored graph.}, 
where all $x_{uv}$ are symbols.

\begin{equation}\nonumber
T_{uv} = 
    \begin{cases}
          x_{uv}, &\text{if  $(u,v)\in E$ and $u>v$ }\\
          -x_{uv}, &\text{if  $(u,v)\in E$ and $u<v$ }\\
          0,  &\text{otherwise}
    \end{cases}
\end{equation}
\end{definition}

The most important application of Tutte Matrix is perfect matching.

\begin{theorem} {\rm \cite{tutte1947factorization}}
\label{theorem:tutte}
A graph $G$ has a perfect matching iff the determinant of its Tutte matrix is non-zero, i.e.,  $det(T)\not\equiv 0$.
\end{theorem}

Symbolic determinant is hard to compute as there can be $O(n!)$ terms in $det(T)$. However, since by Theorem \ref{theorem:tutte} one only needs to know whether $det(T)$ is non-zero to identify existence of perfect matching, polynomial identity testing (PIT) can be applied to achieve this goal.

\begin{lemma} {\rm (Schwartz-Zippel Lemma for PIT) \cite{saxena2009progress}} \label{prop:constant_testing} Let $F:\mathbb{R}^n\to \mathbb{R}$ be a non-zero polynomial and $S\subset \mathbb{R}$ be a finite set with $|S|>n$.  If we pick a point $a=(a_{1},\cdots, a_{n})$, where each element is sampled independently and uniformly at random from $S$, then we have
	$\mathop{\mathbb{P}}\limits_{a\sim \mathcal{U}[S]^n}[F(a)=0]\le \frac{n}{|S|}.$
\end{lemma}

There are $|E|$ variables ($x_{uv}$) in $det(T)$ as a polynomial. Therefore, one can use an integer set with size $c\cdot |E|$ in PIT, where $c$ is a non-zero constant, to guarantee that if a PM exists, the probability of identifying it successfully by PIT is at least $1-1/c$. Computing the determinant of an $n\times n$ numerical matrix can be done in $O(n^3)$ by Gaussian Elimination, or in $O(n^{2.37286})$ by the best-known matrix multiplication algorithm \cite{alman2021refined}. %Thus the decision problem of PM in traditional uncolored graphs is in RP.

\subsubsection{Adapting Tutte Matrix for Bi-Colored Graphs}
  Generalizing  the symbolic-determinant-based approach for XPM \cite{mulmuley1987matching},  we adapt Tutte Matrix to non-simple graphs with bi-colored edges, denoted by $A$ in Definition \ref{defn:matrixA}. Besides $x_{uv}$, for each color $i\in \{1,\cdots,d\}$, we introduce a symbol $y_i$.

\begin{definition} {\rm (Adapted Tutte Matrix)} Given a graph $G=(V,E,d)$ with bi-colored edges, we define the following matrix $A$. For two vertices $u$, $v$, $E_{u,v}$ denotes the set of all edges connecting $u$ and $v$.

\begin{equation}\nonumber
A_{uv} = 
    \begin{cases}
         x_{uv}\cdot  \sum\limits_{e\in E_{u,v}} y_{c_u^e}\cdot y_{c_v^e}, &\text{if  $u>v$ }\\
          - x_{uv}\cdot  \sum\limits_{e\in E_{u,v}} y_{c_u^e}\cdot y_{c_v^e}, &\text{if  $u<v$ }\\
          0,  &\text{otherwise}
    \end{cases}
\end{equation}
\label{defn:matrixA}
\end{definition}

 For example, in the graph shown in Figure \ref{fig:graph}, we have  $A_{v_3v_4}=-(y_{r}^2+y_{r}y_{b})\cdot x_{v_3v_4}.$
 
    Theorem \ref{theo:PIT} shows that the approach based on PIT and determinant in Sec. \ref{subsubsec:tutte}  can be adapted to bi-colored graphs by matrix $A$. The key observation is that legal PMs w.r.t. symmetric vertex-color constraints correspond to certain terms in the polynomial representations of $\sqrt{det(A)}$, or  the \emph{Pfaffian} \cite{galluccio1999theory} of $A$. 
 
\begin{definition} {\rm (Legal terms)}
Let $\mathcal{C}$ be a set of vertex colorings defined by symmetric constraints, the set of \emph{legal terms} is defined by
$
\mathcal{M}_\mathcal{C}=\{\prod_{i=1}^d (y_i)^{ \texttt{count}(c,i) }| c\in \mathcal{C}\}.
$
\end{definition}

For the example in Figure \ref{fig:dd}, $\mathcal{M}_\mathcal{C}=\{y^4_{r},y^3_ry_b,y^3_ry_g,y^2_{r}y_{b}^2,y^2_{r}y_{g}^2,y^2_{r}y_{b}y_{g}\}$. Then, analogous to Theorem \ref{theorem:tutte}, we have the following results. 

\begin{lemma}
A graph has a perfect matching in $\mathcal{C}$ defined by symmetric constraints iff $\sqrt{det(A)}$ has a term in $\mathcal{M}_\mathcal{C}$ with non-zero coefficient.
\label{lemma:PIT}
\end{lemma}

\begin{theorem}
The decision problem of EXISTS-PMVC-Sym with $d$ colors can be solved in $O( \texttt{DET}(n,d) \cdot \log \frac{1}{\epsilon})$ with probability $1-\epsilon$, where $\texttt{DET}(n,d)$ is the complexity of computing the symbolic determinant of an $n\times n$ matrix with $d$ symbols.
\label{theo:PIT}
\end{theorem}

 Symbolic determinant with a bounded number of symbols was shown to be in NC$^2$ \cite{borodin1983parallel}, indicating again that EXISTS-PMVC-Sym-Bounded is in RNC$^2$. Nevertheless, due to the potential multiplication and division of rational functions, bounding the complexity of $\texttt{DET}(n,d)$ by a relatively simple algorithm is non-trivial. We give an upper bound for the sequential complexity of $\texttt{DET}(n,d)$ by fraction-free Gaussian Elimination \cite{bareiss1968sylvester}.

\begin{lemma}
The symbolic polynomial representation of the  determinant of an $n\times n$ matrix with $d$ symbols can be solved in $O(n^{2d+1})$ by fraction-free Gaussian Elimination.
\label{lemma:fractionFree}
\end{lemma} 

It is  not yet obvious how to actually obtain such a PM, which is often desirable.  An elegant solution was proposed in \cite{mulmuley1987matching} for obtaining a PM in XPM, based on isolating lemma. We extend this method for obtaining a legal PM in EXISTS-PMVC-Sym by Theorem \ref{theorem:findingPM}. 
\begin{theorem}
\label{theorem:findingPM} There exists a polynomial randomized algorithm such that if a graph has a PM in the  set of legal vertex colorings $\mathcal{C}$ defined by a system of symmetric constraints, it returns a legal perfect matching with probability at least $1/2$. 
\end{theorem}

The intuition of the isolating-lemma-based algorithm is, when we randomly assign integer weights to all edges, there is a decent chance that there will be a unique legal perfect matching with minimum weight, which the search procedure can focus on.

    On planar graphs, the RNC algorithm for EXISTS-PMVC-Sym-Bounded based on symbolic determinant can be derandomized by \textit{Pfaffian orientation} \cite{thomas2006survey}.  A Pfaffian orientation is an assignment of $\{\pm 1\}$ to all symbols $\{x_{ij}\}$ in the Tutte Matrix $T$, such that all terms in $det(T)$ are positive. The high-level idea is, after we evaluate the matrix $A$ on a Pfaffian orientation deterministically,  $G$ has a legal perfect matching iff $\sqrt{det(A)}$ has a legal term in $\mathcal{M}_\mathcal{C}$ with non-zero coefficient.  It is shown in  \cite{kasteleyn1967graph} that planar graphs are Pfaffian orientable.
    
    \begin{proposition} 
    EXISTS-PMVC-Sym-Bounded is in NC$^2$ on planar graphs.
    \label{prop:pfaffian}
    \end{proposition}

\subsection{A Dynamic Programming Algorithm for Bounded-Treewidth Graphs}
\label{subsec:dp}
 Besides planar graphs, we show that EXISTS-PMVC-Sym-Bounded can also be solved deterministically in polynomial time by dynamic programming (DP) on graphs with bounded treewidth. It is known that DP provides an FPT algorithm for XPM \cite{elmaalouly2022exact}, though the explicit algorithm was not given in previous literature, to our best knowledge. In this paper, we propose a DP algorithm for EXISTS-PMVC-Sym-Bounded as well as its complexity analysis\footnote{In the appendix, we show an attempt of expressing EXISTS-PMVC-Sym by cardinality monadic second-order logic (CMSO).}. 

Treewidth is a graph metric indicating how close a graph is to a tree. Graphs with bounded treewidth naturally emerge in real-life applications, e.g., series/parallel graphs \cite{seriesParallel}, circuit design \cite{wang2001using}, and machine learning \cite{narasimhan2012pac}. Treewidth is of algorithmic interest since many graph problems can be solved in PTIME on bounded treewidth graphs by DP \cite{bodlaender1988dynamic}. 

\begin{definition} {\rm (Tree decomposition and treewidth) }
Let G be a graph. Let $\mathcal{T}$ be a tree where each node of $\mathcal{T}$ is assigned a set of vertices (``bag'') of $G$. We call $\mathcal{T}$ a \emph{tree decomposition} of $G$ if the following holds:
\begin{bracketenumerate}
    \item Each edge in $E$ is contained in a bag; that is, if there is an edge connecting  vertices $u,v\in V$, there exists a node of $\mathcal{T}$ whose bag contains both $u$ and $v$.
    \item The subtree of $\mathcal{T}$ consisting of all bags which contain $u$ is connected, for all  $u\in V$.
\end{bracketenumerate}
The width of a tree decomposition $\mathcal{T}$ is defined as (size of largest bag of $\mathcal{T}$) $-1$. The treewidth of a graph is the minimum width of all its tree decompositions. The treewidth of a tree is $1$.
\end{definition}

\begin{theorem}
\label{theorem:DP}
EXISTS-PMVC-Sym on graphs with treewidth $\texttt{tw}$ can be solved in $O(n^{2d} \cdot (2d+2)^{\tw+1})$, where $d$ is the number of colors.
\end{theorem}

\begin{proof} We first introduce a special type of tree decomposition which is easier to work with. 
\begin{definition} {\rm (Nice tree decomposition)}
A nice tree decomposition is a tree decomposition where every node has one of the following four types:
\begin{bracketenumerate}
     \item A \textbf{leaf node} has no children and its bag has one vertex.
     \item An \textbf{introduce node} has one child. The child has a bag of same vertices as the parent with one vertex deleted.
     \item A \textbf{forget node} has one child. The child has a bag of same vertices as the parent with one vertex added.
     \item A \textbf{join node} has two children, both have a bag identical to the parent's.
\end{bracketenumerate}
\end{definition}

A nice tree decomposition can be obtained from a tree decomposition with polynomial overhead, as Proposition \ref{prop:numBags} states.

\begin{proposition} {\rm \cite{kloks1994treewidth}}
A nice tree decomposition with treewidth $\tw$ can be converted from a tree decomposition with treewidth $\tw$ in $O(\tw^2\cdot n)$ time with $O(\tw \cdot n)$ bags.
\label{prop:numBags}
\end{proposition}

In the following, we define indicator variables for memoization in our DP algorithm. For a node $X$ in the tree decomposition $\mathcal{T}$,  let \texttt{tree}$(X)$ be the set of nodes of the subtree in $\mathcal{T}$ with root $X$. Let $\bag{X}$ be the set of vertices in $G$ contained in $X$. For a set $S$ of nodes in $\mathcal{T}$, let $\bag{S}$ be $\cup_{X\in S}\bag{X}$.

\begin{definition}
{\rm (Partial vertex coloring)} A partial vertex coloring of a vertex set $V'\subseteq V$ is a mapping from $V'$ to $\{0,1,\cdots, d\}$. $c(v)=0$ indicates that $v$ is not matched yet, while $c(v)=i\in \{1,\cdots, d\}$ means $v$ is matched in the graph with inherited color $i$. A partial vertex coloring $c$ of $V'$ is \emph{feasible} if there is a matching in the induced graph of $G$ on $V'$ whose inherited partial vertex coloring is $c$. %I.e., there is a matching in $G$ that matches all $v$ such that $c(v)\in\{1, \cdots, d\}$ whose inherited vertex coloring is the color on all vertices.
\end{definition}

An indicator Boolean variable in our DP-based algorithm are defined as $I(X,c,b)$, where $X$ is a node in the nice tree decomposition, $c:V\to\{0,1,\cdots,d\}$ is a partial vertex coloring on $\bag{X}$ and $b:\{1,\cdots,d\}\to \mathbb{N}$ is a count of colors.  Then we define $I(X,c,b)$ to be \texttt{true} if the partial vertex coloring $c$ of $\bag{X}$ can be extended to a \emph{feasible} partial coloring $c'$ of all vertices in  $\bag{\texttt{tree}(X)}$, such that each color $i\in \{1,\cdots,d\}$ appears exactly $b(i)$ times in $c'$. Otherwise, $I(X,c,b)$ is defined to be \false. In our DP-based approach, $I(X,c,b)$ can be computed recursively as follows:

\begin{bracketenumerate}
\item \textbf{$X$ is a leaf.} Suppose $\bag{X}=v$. Set $I(X,\{c(v)=0\},0)$ to \true. Set $I(X,c,b)$ to \false\, for all other $c$ and $b$ since the only one node must be unmatched.

\item \textbf{$X$ introduces $v$ from its child $Y$.} 
\begin{itemize}
    \item If $c(v)=0$, i.e., $v$ is not matched in the induced subgraph on $\bag{\texttt{tree}(X)}$, then $I(X,c,b)=I(Y, c', b)$, where $c'(u)=c(u)$ for all $u\in\bag{Y}$. 
    \item If $c(v)\neq 0$, i.e., $v$ is matched with some vertices in $\bag{\texttt{tree}(X)}$. By Proposition \ref{prop:introduceNode}, $v\not\in\bag{\texttt{tree}(Y)}$. Therefore, by the definition of tree decomposition, $\bag{Y}$ must contain all $v$'s neighbours in $\bag{\texttt{tree}(X)}$. Thus $v$ must be matched with a vertex $u\in\bag{Y}$. Therefore we have $I(X,c,b)=\bigvee_{u\in \bag{Y},c(u)\neq 0, \{(u, c(u)), (v,c(v))\} \in E} I(Y,c',b')$, where  $c'(u)=0$ and $c'(w)=c(w)$ for all $w\in \bag{Y}-u$;  $b'(c(u))=b(c(u))-1$, $b'(c(v))=b(c(v))-1$ and $b'(i)=b(i)$ for $i\not\in \{c(v),c(u)\}$.
\end{itemize}

\item \textbf{$X$ forgets $v$ from its child $Y$.} Since $v$ is forgotten in $X$, it must have already been matched in $\bag{\texttt{tree}(Y)}$. Therefore, we have
$
I(X,c,b)=\bigvee_{c'}I(Y,c',b)
$
for all $c'$ s.t. $c'(v)\neq 0$ and $c'(w)=c(w)$ for all $w\in \bag{X}$.

\item \textbf{X joins its children $Y$ and $Z$.} By Proposition \ref{prop:joinNode}, $b$ can be split into two color counts corresponding to $\texttt{tree}(Y)$ and $\texttt{tree}(Z)$. Hence we have
$
I(X,c,b)=\bigvee_{b_Y,b_Z,c_Y,c_Z} \big( I(Y,c_Y,b_Y) \wedge I(Z, c_Z, b_Z) \big)
$
for all 1) $b_Y, b_Z:\{1,\cdots,d\}\to\mathbb{N}$, s.t. $b_Y(i)+b_Z(i)=b(i)+\texttt{count}(c,i)$; 2) $c_Y,c_Z$ such that for all $v\in \bag{X}$,  $c_Y(v)=c(v)$, $c_Z(v)=0$ or $c_Z(v)=c(v)$, $c_Y(v)=0$. 
\end{bracketenumerate}
Finally, graph $G$ has a legal PM w.r.t. symmetric constraints iff. $\bigvee_{b_l,c}I(\mathcal{T}.\texttt{root},c,b_{l})=\true$, for all $c:\bag{\mathcal{T}.\texttt{root}}\to\{1,\cdots,d\}$ and all legal count $b_l$ w.r.t. the symmetric constraints (note that we have $\sum_{i\in\{1,\cdots,d\}}b_l(i)=n$).

The analysis of complexity can be founded in the appendix.
\end{proof}

\begin{corollary}
\label{coro:xPMBoundedTW}
The DP-based algorithm solves XPM on bounded-treewidth graphs  deterministically in $O(n^4\cdot6^
\tw)$.
\end{corollary}

\iffalse
\section{Experiment Results}
This section provides an evaluation of different approaches to solving EXISTS-PMVC problems.

\subsection{Implementations and Methods in Comparison}
We implemented the following methods. More details regarding the implementations of 
\begin{itemize}
    \item Symbolic determinant. We implemented the methods described in the section.
    \item Dynamic programming described in section .
    \item enumeration and blossom algorithm described by proposition x.
    \item A propositional logic approach. It is well known that perfect matching can be expressed by exact-one constraints.
\end{itemize}

We generate benchmarks with graphs that satisfies Dicke state and GHZ state. Each graph has n from 20 to 1000. Those are graphs that would return "False". The results can be found in Figure 1 and 2.

We also generate satisfiable instances, with adding 1 random edges to the graphs such that the properties are falsified. The results can be found in Figure 3 and 4.

Things to do
\begin{itemize}
    \item generate graph with GHZ and Dicke state and store them.
    \item randomly generate 1 edge for each graph and store them.
    \item Run experiments for them
\end{itemize}

\fi

\section{Conclusions and Future Directions}
\label{sec:conlusion}
In this paper, we proposed and studied the problem of perfect matching under vertex-color constraints (EXISTS-PMVC), with motivation in quantum experiment design. We reveal that the complexity of EXISTS-PMVC heavily relies on the representation of the constraints. We first prove that if the constraints are allowed to be expressed in decision diagrams, the problem is NP-hard by a graph-gadget technique. Then we connect EXISTS-PMVC with symmetric vertex-color constraints with a well-known combinatorial problem, exact perfect matching (XPM) by showing those two problems can be polynomially reduced to one another. We give algorithms for EXISTS-PMVC-Sym that natively handle bi-colored graphs. Our results shed light on understanding the complexity of testing quantum states of large quantum optical circuits. 

Our work leaves interesting open questions and future directions. For example, on what special classes of graphs can EXISTS-PMVC be solved efficiently? Are there more interesting types of vertex-color constraints? Motivated by leveraging classical algorithms in the development of quantum computing, we plan to implement scalable algorithms for EXISTS-PMVC and matching partition functions \cite{barvinok2011computing}. Besides PM, we also expect to tackle other quantum-inspired classical problems by combinatorial and logical approaches. 
%%
%% Bibliography
%%

%% Please use bibtex,
\bibliography{lipics-v2021-sample-article}

\appendix
\section{Appendix}\label{sec:itemStyles}

\subsection{Technical Proofs}
\subsubsection{Proof of Proposition \ref{prop:polynomiaClIsEasy}}
When $|\mathcal{C}|$ is polynomial, EXISTS-PMVC can be solved by calling the following procedure for each vertex coloring $c\in \mathcal{C}$.

Given a vertex color $c$, let $G_c$ be the subgraph with the original vertex set $V$ but a subset of edges, selected as follows:
$$
E_c=\{c_u^e = c(u),c_v^e=c(v)|\{(u,c_u^e),(v,c_v^e)\}\in E\}
$$
That is, in $G_c$, we only keep the edges whose colors agree with the vertex coloring $c$. Then we run a polynomial algorithm $\mathcal{A}$ for perfect matching on uncolored graphs, say Blossom Algorithm. In the following we prove that the original graph $G$ has a perfect matching whose inherited vertex coloring is $c$ iff. $\mathcal{A}$ returns \true.

$\Rightarrow:$ If there is a PM $P$ whose inherited vertex coloring is $c$, then every edge of $P$ must be in $E_c$. Therefore $G_c$ has an uncolored perfect matching and $\mathcal{A}$ returns \true.

$\Leftarrow:$ Suppose $\mathcal{A}$ returns \true, which means $G_c$ has an uncolored PM $P$. Then the colors of every edge in $P$ agree with $c$, which indicates $G$ has a perfect matching with inherited color $c$. 
 
Therefore, the answer for the decision version of EXISTS-PMVC is \true\, iff. $\mathcal{A}$ returns \true\, for at least one vertex coloring $c\in\mathcal{C}$. 
\hfill\qedsymbol

\subsubsection{Proof of Theorem \ref{theo:EXISTS-PMVC-BDD}}

\begin{proof}
We reduce 3-SAT to EXISTS-PMVC-DD with only red-blue monochromatic edges by the gadget technique, see Figure \ref{fig:3CNF} as an illustration. Suppose $f=C_1\wedge \cdots\wedge C_m$ is an arbitrary 3-CNF formula with variable set $X$, where each $C_i (1\le i\le m$) is a clause with $3$ literals.

\noindent\textbf{Graph construction:} For each clause $C_i=l_{i,1}\vee l_{i,2}\vee l_{i,3}$, we construct a gadget consisting of 6 vertices: 1 vertex $u_i$ representing the clause, 3 vertices $v_{i,1}$, $v_{i,2}$, $v_{i,3}$ representing the three literals in the clause and 2 dummy vertices $w_{i,1}$ and $w_{i,2}$. 

\begin{figure}[h]
	\centering
\includegraphics[scale=0.23]{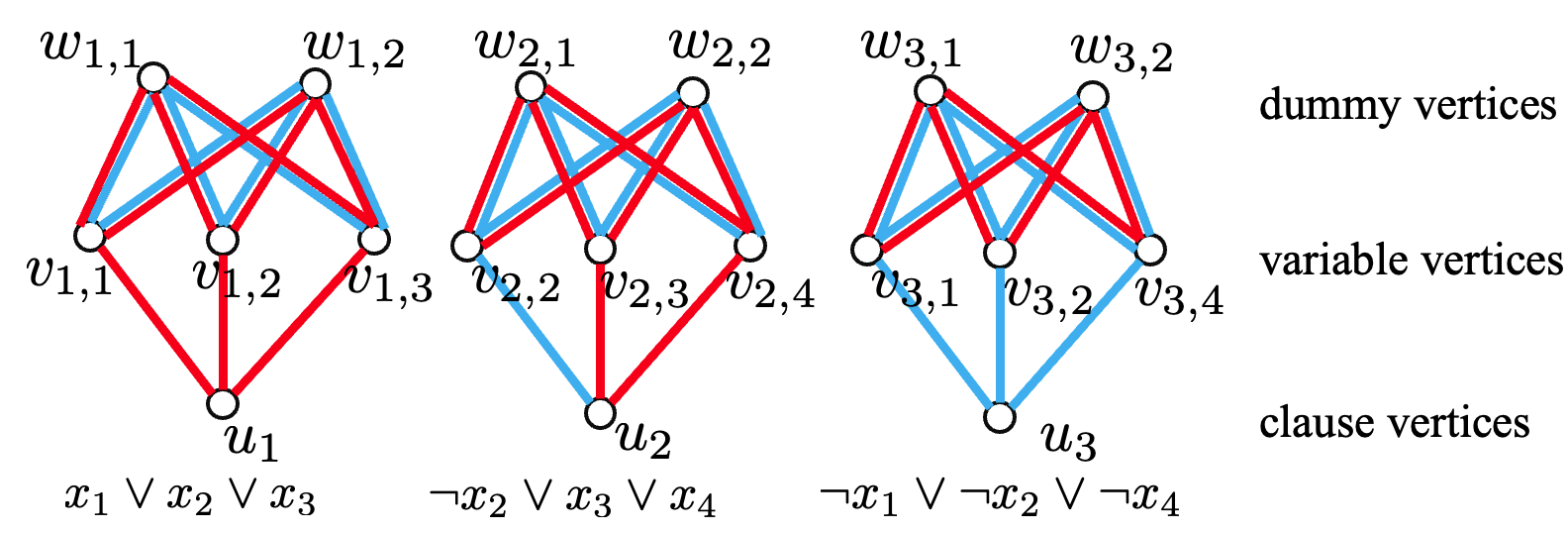}
\caption{An example of reducing a 3-CNF formula to a graph in EXISTS-PMVC-DD.}
\label{fig:3CNF}
\end{figure}

 There is an edge that connects each vertex $v_{i,k}$ ($k=1,2,3$) corresponds to a literal with the vertex $u_i$ corresponding to clause $C_i$. Each edge between $v_{i,k}$ to $u_i$, is colored red if the literal $l_{i,k}$ is positive in $C_i$. Otherwise the edge is colored blue. There is one blue and one red edge connecting each dummy vertex $w_{i,k}$ ($k=1,2$) and each vertex representing the literal $v_{i,k}$ ($k=1,2,3$). The graph $G$ is the union of gadgets corresponding to all clauses. 

\noindent\textbf{DD construction:} For each variable $x\in X$, we denote by $V^x$ the set of all vertices of $G$ that corresponds to $x$. The DD $D$ in our reduction is the conjunction of DDs of the mono-chromatic constraint for each variable $x\in X$, i.e., all vertices in $V_x$ have the same color. Since there are only two colors in our reduction, a binary DD (BDD) is constructed. 

\noindent\textbf{Correctness of the reduction:} we need to prove:
\begin{bracketenumerate}
\item  The graph $G$ has a perfect matching with inherited vertex coloring satisfying the binary DD $D$ iff. the 3-CNF formula $f$ is satisfiable.
\item The reduction is polynomial. The graph $G$ is with polynomial size since it has $6\cdot m$ vertices where $m$ is \# clauses of $f$. We need to show that there exists a binary DD $D$ with polynomial size that represents the conjunction of color constraints. Moreover, $D$ needs to be constructed in polynomial time.
\end{bracketenumerate}

\noindent\textbf{Proof of (1)}:
\noindent$\Leftarrow$: Suppose the 3-CNF formula is satisfiable and $s:X\to \{\true,\false\}$ is a solution. We construct a legal perfect matching $P$ w.r.t. $D$. For each clause $C_i$, there is at least one literal whose assignment in $s$ satisfies the clause. Without loss of generality, suppose $l_{i,1}$ satisfies the clause. Then we add the edge between $v_{i,1}$ to $u_i$ into the matching $P$. By the construction of the graph $G$, the color of this edge is red if $l_{i,1}$ is assigned \true\, in $s$ and blue if $l_{i,1}$ is assigned \false\, in $s$.  For the two vertices corresponding the rest two literals in the clause, i.e., $v_{i,2}$ and $v_{i,3}$, we match them with two dummy vertices $w_{i,1}$ and $w_{i,2}$. For example, we add to $P$ an edge between $v_{i,2}$ and $w_{i,1}$, whose color is red if the variable of $l_{i,2}$ is assigned \true\,  in $s$ and blue if $l_{i,2}$ is assigned \false\, in $s$.  We add an edge to $P$ between $v_{i,3}$ and $w_{i,2}$ similarly. All 6 vertices in the gadget of clause $i$ is perfectly matched. Therefore $P$ is a perfect matching of $G$. We need to prove that our perfect matching $P$ is legal w.r.t. the DD $D$. Due to the construction of $P$, for each variable $x$, the inherited color of vertices in $V^x$ w.r.t. $P$ must be the same, which is determined by the unique assignment of $x$ in the solution $s$. This is exactly what $D$ encodes. Hence $P$ is legal w.r.t. $D$.

\noindent$\mathbf{\Rightarrow}$: Suppose graph $G$ has a legal perfect matching $P$ w.r.t. $D$. Then we construct an assignment $s:X\to\{\true,\false\}$ from $P$ and show that $s$ is a solution of $f$. Since $P$ is legal w.r.t. $D$, the inherited color of all vertices in $V^x$ must be the same for all $x\in X$. For each variable $x$, $s(x)=\true$ if the inherited color of all vertices in $V^x$ w.r.t. $P$ is red and $s(x)=\false$ otherwise.  We show that $s$ is a solution. For each clause $C_i$, since its corresponding gadget is perfectly matched in $P$, the vertex $u_i$ must be connected with a vertex corresponding to one of the three literals, say $l_{i,1}$. This implies that literal $l_{i,1}$ is satisfied by $s$ and so is clause $C_i$. Hence $s$ is a solution of formula $f$ and $f$ is satisfiable. \hfill\qedsymbol

\noindent\textbf{Proof of (2):}
$D$ is the conjunction of  ``the color of all vertices in $V_x$ are equal'' for all $x\in X$, which is an \texttt{All-Equal} constraint. It is well known that as a symmetric Boolean function, an \texttt{All-Equal} constraint admits a polynomial size DD \cite{kyrillidis2021continuous}. We denote the DD of the \texttt{All-Equal} constraint regarding $x$ as $D_x$. In addition, note that for two variables $x_1,x_2\in X$, we have $V_{x_1}\cap V_{x_2}=\emptyset$. A DD representing $D_{x_1} \wedge D_{x_2}$  with size $O(|D_{x_1}|+|D_{x_2}|)$ can be constructed in polynomial time. This is done simply by substituting $D_{x_1}.\true$ by $D_{x_2}$ (note that $D_{x_1}$ and $D_{x_2}$ do not share variables). Thus a DD representing  $\bigwedge_{x\in X}D_x$ with polynomial size can be constructed  in polynomial time and the reduction is polynomial.
\end{proof}

\subsubsection{Proof of Lemma \ref{lemma:PIT}}
We use the following result of Pfaffian of a skew-symmetry matrix:

Let the set of partitions of $(1,n)$ into pairs without regard with order as $\{(i_1,j_1),\cdots, (i_{n/2},j_{n/2})\}$. 
$$
\texttt{Pf}(A) = \sqrt{Det(A)}=\sum_{p\in P}  A_{i_1,j_1} \cdot  A_{i_2,j_2} \cdots  A_{i_{n/2},j_{n/2}}.
$$

$\Leftarrow:$ Suppose $G$ has a perfect matching $P$ whose inherited vertex coloring $c$ is in $S$.  Given perfect matching $P$ with each edge and
(arbitrarily) ordered pair $(u,v)$. Set $x_{uv} = 1$ if $\{(u, c(u)), (v, c(v))\}\in P$ and $u > v$. Otherwise set $x_{uv} = 0$. Let $\pi$ be a permutation such that for each $\{(u, c(u)), (v, c(v))\}\in P$ we have $u = \pi(v)$ and $v = \pi(u)$. Now consider the term $\prod_{i=1} A_{i\pi(i)}$ of $\texttt{Pf}(A)$. It is clearly equal to $\prod_{i}y_i^{c(i)}$, whose coefficient is non-zero. 

\noindent$\Rightarrow:$  Suppose $\texttt{Pf}(A)$ has a non-zero term in $M$, by the expression of Pfaffian, there is a partition of all vertices into pairs such that all $A_{i_1,j_1}$ is non-zero, which yields a PM. Moreover,  each color $i$ appears in the inherit vertex coloring of this matching exactly $c(i)$ times.
\hfill\qedsymbol

\subsubsection{Proof of Theorem \ref{theo:PIT}}
The following is an RP algorithm for EXISTS-PMVC-Sym:
\begin{bracketenumerate}
\item Construct the symbolic matrix $A$ from graph $G$ by Definition \ref{defn:matrixA}.
\item Choose the set size $2|E|$ and conduct PIT on $\texttt{Pf}(A) = \sqrt{det(A)}$.
\item If there exists a term in the polynomial of $det(A)$ in $\mathcal{M}$ with non-zero coefficient, return ``yes'', otherwise return ``no''.
\end{bracketenumerate}
If $G$ has no legal perfect matching, then after PIT in (2), (3) always returns ``no''. If there exists at least one legal PM, then the algorithm above gives ``yes'' with probability at least $1/2$. Therefore by $O(\log\frac{1}{\epsilon})$ independent runs of this algorithm, we get the correct answer to EXISTS-PMVC-Sym with probability at least $1-\epsilon$.
\hfill\qedsymbol

\subsubsection{Proof of Lemma \ref{lemma:fractionFree}}

The fraction-free Gaussian Elimination \cite{bareiss1968sylvester} for computing the determinant of a matrix $a$ computes a sequence of matrix $a^{k}$, where $a^1=a$ and 

$$
a_{ij}^{k+1}=\frac{a_{kk}^{k}a_{ij}^k-a_{ik}^ka_{kj}^k}{a^{k-1}_{k-1,k-1}},\,\,\,\, i,j\in\{k+1,\cdots,n\}\times \{k,\cdots,n\}
$$
and $a_{00}^0$ is defined to be $0$. After all the computations, $a_{nn}^n$ is the determinant of $a$.

It is proven that each term computed by the equation above is fraction-free, i.e., $a^{k-1}_{k-1,k-1}$ is a factor of $a_{kk}^{k}a_{ij}^k-a_{ik}^ka_{kj}^k$. The procedure requires $O(n^3)$ polynomial multiplications and divisions \cite{khovanova2013efficient}. Multiplying or devising two polynomials $p,q$ takes $O(T(p)T(q))$, where $T(p)$ is the number of terms in $p$.

\begin{lemma}
Let $A$ be the matrix defined in Definition \ref{defn:matrixA} and $A_{ij}^k$ be an element computed in fraction-free Gaussian Elimination of $A$. Then $A_{ij}^k$ is either $0$ or homogeneous of degree $2k$. 
\end{lemma}
\begin{proof}
We prove by induction.

Basis: when $k=1$ the statement holds since every term in $A$ is either $0$ or a homogeneous polynomial with degree $2$. 

Inductive step: suppose each $A_{ij}^k$ ($k\ge 1$) is either $0$ or a homogeneous polynomial of degree $2k$. Then 
$$
A_{ij}^{k+1}=\frac{A_{kk}^{k}A_{ij}^k-A_{ik}^kA_{kj}^k}{A^{k-1}_{k-1,k-1}}
$$
By inductive hypothesis, $A_{kk}^{k}A_{ij}^k-A_{ik}^kA_{kj}^k$ is either $0$ or homogeneous with degree $2k+ 2k=4k$. It was proved in \cite{turner1995gauss} that $A^k_{kk}$ is the determinant of the primal minor of $A$ of dimension $k$ for all $k$.  Therefore $A^{k-1}_{k-1,k-1}$ is non-zero and must be homogeneous of degree $2k-2$, by inductive hypothesis. Since $A_{k-1}A_{k-1}^{k-1}$ divides $A_{kk}^{k}A_{ij}^k-A_{ik}^kA_{kj}^k$, we have $A_{ij}^{k+1}$ is either $0$ or homogeneous of degree $4k-(2k-2)=2(k+1)$, which completes the proof.
\end{proof}

Back to the proof of Lemma \ref{lemma:fractionFree}, all polynomials during the fraction-free Gaussian elimination are either $0$ or homogeneous polynomials of degree $2n$. Therefore, the cost of multiplying or devising two polynomials is bounded by $O(n^{2d-2})$. Hence the complexity of symbolic determinant of $A$ by fraction-free Gaussian Elimination is bounded by $O(n^{2d-2}\cdot n^3)=O(n^{2d+1})$.
\hfill\qedsymbol

\subsubsection{Proof of Theorem \ref{theorem:findingPM}
(Obtaining a Perfect Matching)}

Some proofs in this section are adapted from \cite{lectureNotes} for handling EXISTS-PMVC-Sym.
\begin{definition} {\rm (Set system)}
A set system $(S, F)$ consists of a
finite set S of elements,
$s = \{x_1, x_2,\cdots, x_n\}$, and a family $F$ of subsets of $S$, i.e. $F = {S_1,S_2,\cdots,S_k}$, $S_j\subseteq S$,
for $1 \le j \le k$.
\end{definition}

Let us assign a weight $w_i$ to each element $X_i\in S$ and let us define the weight of the set $S_j$ to be $\sum_{x_i\in S_j} w_i$. Then the isolation lemma is as follows:

\begin{lemma}
{\rm (Isolating Lemma) \cite{valiant1985np}}  Let $(S, F)$ be a set system whose
elements are assigned integer weights chosen
uniformly and independently from $\{1,2,\cdots, K\}$,
Then, 
$$\mathbb{P} [\text{There is a unique minimum weight set in $F$}] \ge \frac{n}{K}$$
\end{lemma}

In our case, the set system consists of $|E|$ edges as elements, while all perfect matchings that satisfy the symmetric constraints form the system. %Instead of using an arbitrary set of numbers for solving the decision problem, now we need to use a specific set of integers. 
We assign random integer weights from $\{1,\cdots, 2\cdot|E|\}$ to all edges. Suppose the weight assigned to edge $(i,j)$ is $w_{ij}$. We have the following modification to the matrix $A$ to get matrix $B$:

\begin{equation}\nonumber
B_{uv} = 
    \begin{cases}
           \sum\limits_{e\in E_{u,v}} y_{c_u^e}\cdot y_{c_v^e}\cdot 2^{w_{e}}, &\text{if  $u>v$ }\\
        \sum\limits_{e\in E_{u,v}} y_{c_u^e}\cdot y_{c_v^e} \cdot ( - 2^{w_{e}}), &\text{if  $u<v$ }\\
          0,  &\text{otherwise}
    \end{cases}
\end{equation}

Let the weight of each perfect matching $P$ be $\sum_{e\in P}w_e$.

\begin{lemma} If graph $G$ has a legal perfect matching, then with probability at least $\frac{1}{2}$  there exists a legal term in $\texttt{Pf}(B) = \sqrt{det(B)}$ with non-zero coefficient $p$. The highest power of $2$ that divides $p$ is $2^{2W}$, where $W$ is the weight of the legal perfect matching with minimum weight.
\end{lemma}
\begin{proof}
Recall that terms corresponding to permutations with odd cycles cancel out. For the permutation $\sigma$ corresponding to the unique legal minimum-weight perfect matching in $G$, we have $\prod_{i=1}B_{i,\sigma(i)}=\pm 2^{2W}$ . We will show that for all other permutations $\tau$, we have $\prod_{i=1} T_{i,\tau(i)} = \pm 2^{2W'}$ for some integer $a > W$, which completes the proof.
Let $\tau$ be a permutation not corresponding to the unique legal minimum-weight perfect matching in $G$. If $\tau$ corresponds to a legal matching of weight $W'$, then $\prod_{i=1} B_{i,\tau(i)} = \pm2^{2W'}$. Since this matching is
not the one of minimum-weight, we have $W > W'$ as required.
Suppose instead that $\tau$ is a permutation consisting solely of even cycles, but not necessarily corresponding to a matching in $G$. Suppose that $\prod_{i=1} B_{i,\tau(i)} = \pm2^{2W'}$ with $W'\le W$. Then we show that a legal perfect matching with weight $W'$ can be constructed from $\tau$, which conflicts with the assumption that there is a unique legal perfect matching with weight $W$. For each even circle $\mathcal{C}$, one can construct two matchings that cover all vertices in $\mathcal{C}$. We choose the one with lower weight. The matchings from all even circles forms a legal perfect matching of $G$, with weight $\le W$.
\end{proof}

After obtaining the weight $W$ of the legal PM $P$ with minimum weight, for each edge $e$, whether it is in the PM $P$ can be determined by the following lemma.

\begin{lemma} {\rm \cite{mulmuley1987matching}} For each edge $e\in E$, $e$ belongs to the legal perfect matching iff. $\frac{det(B_{ij})}{2^{W-w_e}}$ is odd.
\end{lemma}
Therefore, a legal perfect matching can be obtained by a randomized algorithm in polynomial time.\hfill\qedsymbol

   \subsubsection{Proof of Proposition \ref{prop:pfaffian}}
    Since the graph $G$ is planar, we can evaluate the matrix $A$ in Definition \ref{defn:matrixA} at the Pfaffian Orientation of $G$ \cite{kasteleyn1967graph}. The Pfaffian Orientation evaluates every term in $det(A)$ to be positive. By similar argument in the proof of Lemma \ref{lemma:PIT}, graph $G$ has a legal PM iff. the evaluation of $det(A)$ at the Pfaffian Orientation has a term in $\mathcal{M}_\mathcal{C}$. Since the evaluation of symbolic determinant is in NC$^2$ \cite{borodin1983parallel},  EXISTS-PMVC-Sym on planar graphs is also in NC$^2$.
\hfill\qedsymbol

\subsubsection{Auxiliary Propositions for Proving Theorem \ref{theorem:DP}}
\begin{proposition}
On a node $X$ that introduces $v$ with child $Y$, $v\not\in\texttt{bag}(\texttt{tree}(Y))$.
\label{prop:introduceNode}
\end{proposition}
\begin{proof}
Suppose $v\in \texttt{bag(\texttt{Y})}$. Since $X$ introduces $v$, $v\not \in \texttt{bag}(Y)$. Therefore $v$ must appear in the bag of a node $Z$, which is separated from $X$ by $Y$. This contradicts with the second condition in the definition of tree decomposition.
\end{proof}

\begin{proposition}
On a join node $X$ with children $Y$, $Z$,  $\texttt{bag}(X)=\texttt{bag}(\texttt{Tree}(Y))\cap \texttt{bag}(\texttt{Tree}(Z))$.
\label{prop:joinNode}
\end{proposition}
\begin{proof}
We first prove that for each vertex $v$ in both $\texttt{bag}(\texttt{Tree}(Y))$ and $\texttt{bag}(\texttt{Tree}(Z))$, $v\in \texttt{bag}(X)$. Suppose there exists a vertex $v\in \texttt{bag}(\texttt{Tree}(Y))\cap \texttt{bag}(\texttt{Tree}(Z))$ but $v\not\in \texttt{bag}(X)=\texttt{bag}(Y)=\texttt{bag}(Z)$. Then there exist a node $Y'\in \texttt{Tree}(Y)$ and a node $Z'\in \texttt{Tree}(Y)$ such that $v\in \texttt{bag}(Y')$ and $v\in \texttt{bag}(Z')$. However, $Y'$ and $Z'$ are separated by $X$, which violates the second condition in the definition of nice tree decomposition.

Next we show that if a vertex $v$ is in $\texttt{bag}(X)$, then $v$ must be in $\texttt{bag}(\texttt{Tree}(Y))\cap \texttt{bag}(\texttt{Tree}(Z))$. Since $X$ is a join node, we have $\texttt{bag}(X)=\texttt{bag}(Y)=\texttt{bag}(Z)$ and the statement holds trivially.
\end{proof}

\subsubsection{Complexity Analysis of Theorem \ref{theorem:DP}}
The time complexity of this DP-based procedure can be bounded by the product of 1) the number of all possible states of $I$, denoted by $|I|$ and 2) the maximum number of states that one state relies on.

 The number of states is given by the product of the following three parts:
\begin{alphaenumerate}
    \item number of bags (nodes in the tree decomposition): $O(\tw\cdot n)$ by Proposition \ref{prop:numBags}.
    \item number of possible partial colorings $c$: $(d+1)^{\tw+1}$.
    \item number of possible count of colors $b$: $\binom{n}{d-1}\sim O(n^{d-1})$.
\end{alphaenumerate}
Therefore the number of all possible states $(X,c,b)$ is $O(\tw\cdot n^{d}\cdot (d+1)^{\tw+1})$.

The maximum number of states considered for computing one state is bounded by the product of 1) the number of possible ways of splitting a count of all colors $b$ at a join node $X$ into $b_Y$ and $b_Z$, bounded by  $\prod_{i=1}^d(b(i)+1) \le ((\ceil{\frac{n}{d}}+1)^d)\le  (\frac{n+2d}{d})^d=O(n^{d})$; and 2) the number of possible ways of splitting the coloring to $\bag{X}$ at a join node $X$ into $c_Y$ and $c_Z$, which is bounded by $2^{tw+1}$. Hence the overall complexity of the DP approach is $O(n^{2d}\cdot (2d+2)^{\tw+1})$.

\subsection{CMSO representation of EXISTS-PMVC-Sym}
Not all graph properties on bounded treewidth graphs can be solved efficiently. The seminal meta-theorem of Courcelle reveals the deep connection between monadic second-order logic (MSO) and tractable graph properties on bounded treewidth graphs. We show an attempt of expressing EXISTS-PMVC-Sym by MSO formulas. Nevertheless, the length of the MSO formula is polynomial w.r.t. the size of the graph instead of constant, which means Courcelle's Theorem can not be directly applied to get a linear-time DP-based algorithm for EXISTS-PMVC-Sym-Bounded.
\begin{definition}
{\rm (Extended monadic second-order logic (MSO) on graphs) \cite{courcelle1990monadic}}  MSO is a fragment of second-order logic where the second-order quantification is limited to quantification over sets. The extended MSO allows a graph property to be expressed by $\{V,E,\mathbf{edg},\mathbf{Card},P_1,\cdots,P_k\},$ where $V$ is the set of vertices; $E$ is the set of (directed) edges; $\mathbf{edg}_G(e,u,v)=\true$ iff. there is an edge from $u$ to $v$ in $G$; $\mathbf{Card}_{p,n}(X)\Leftrightarrow |X|\equiv p \mod n$, $X$ is a subset of $V$ or $E$; $P_1,\cdots,P_k$ are subsets of $V$ or $E$, which can be viewed as vertices and edge labels/colors.  

\end{definition}

\begin{theorem}
{\rm (Courcelle's Theorem) \cite{courcelle1990monadic}} Let $f$ be a formula in extended MSO. Then there is an algorithm which, on every graphs with $n$ vertices and treewidth $\tw$, decides whether $f$ holds in $G$ in $O(F(|f|,\tw)\cdot n)$, for some computable function $F$.
\end{theorem}

\begin{proposition}
\label{prop:monadic}
EXISTS-PMVC-Sym can be expressed in extended MSO with polynomial length.
\end{proposition}

We show that EXISTS-PMVC-Sym on bi-colored graphs can be expressed in extended Monadic second-order logic (MSO) with cardinality predicates. To represent possibly two colors on one edge, we view all edges in $G$ as directed edges, while the direction can be arbitrary.  Then, for each color, we define two edge sets $E_{tail}^i$ and $E_{head}^i$ as follows:

 $$E_{head}^i=\{e=(u,v)\in E | \text{the color of $e$ at $u$ is $i$}\},$$
 
 $$E_{tail}^i=\{e=(u,v)\in E | \text{the color of $e$ at $v$ is $i$}\}.$$

For the sake of simplicity, we define the adjacent relations between a vertex and an edge by the predicate $\mathbf{edg}$:

$$head(v,e)=\exists u\in V,  \mathbf{edg}(e,v,u),$$

$$tail(v,e)=\exists u\in V, \mathbf{edg}(e,u,v),$$

$$adj(v,e)=head(v,e)\vee tail(v,e).$$

For two edges, we define a predicate indicating whether two edges share a vertex:
$$
overlap(e_1,e_2)  = \exists v\in V, adj(v,e_1) \wedge adj(v,e_2).
$$

Then the MSO2 formula of EXISTS-PMVC-Sym is:

$$\exists V_1 V_2\cdots V_d\subseteq V, P\subseteq E, \texttt{INDEPENDENT} \wedge \texttt{COLOR} \wedge \texttt{CARD},$$ where

$$\texttt{INDEPENDENT}=\forall e_1,e_2\in P, \neg overlap(e_1,e_2)$$
$$\texttt{COLOR}= (\forall v\in V, \exists e\in P, (adj(v,e))\wedge \big( \bigwedge_i (v\in V_i)\to   \big((head(v,e)\to  e\in E_{head}^i) \wedge(tail(v,e)\to e\in E_{tail}^i)\big) \big)$$

$$\texttt{CARD} =  \bigvee_{c\in \mathcal{C}}\bigwedge_i Card_{c(i),n}(V_i)$$
\hfill\qedsymbol

The size of the CMSO formula above is not constant, but polynomially bounded by the size of the input graph. The polynomial size of the formula come from 1) symmetric constraints encoded by the disjunction of all legal colorings. 2) the unary encoding of moduli in the counting predicates.
\end{document}